\newif\ifreport
\reporttrue

\ifreport
\documentclass{article}
\usepackage{authblk}
\newenvironment{frontmatter}{}{}
\newenvironment{keyword}{Keywords:}{}

\usepackage{a4wide}

\else
\documentclass[review]{elsarticle}
\fi

\usepackage{lineno}

\usepackage[utf8]{inputenc}
\usepackage{amsmath}
\usepackage{etoolbox}

\usepackage{amsfonts}
\usepackage{graphicx}
\usepackage{tikz}
\usetikzlibrary{calc}
\usetikzlibrary{positioning}
\usetikzlibrary{arrows}
\usetikzlibrary{matrix}
\usetikzlibrary{shapes}
\usetikzlibrary{decorations.text}
 
\usepackage{xspace}
\usepackage{tabularx}

\usepackage{makecell}
\usepackage[ruled,vlined]{algorithm2e}
\DontPrintSemicolon

\tikzstyle{robot}=[circle, draw, fill=black!80,inner sep=0pt, minimum width=4pt]

\tikzstyle{dest}=[circle, draw, fill=white,inner sep=0pt, minimum width=4pt]

\newcommand{\R}{\mathbb{R}}
\newcommand{\Z}{\mathbb{Z}}
\newcommand{\N}{\mathbb{N}}
\newcommand{\ie}{\emph{i.e.}\xspace}

\newtheorem{theorem}{Theorem}
\newtheorem{lemma}{Lemma}
\newtheorem{corollary}{Corollary}
\newtheorem{remark}{Remark}

\newenvironment{proof}{{\bf Proof. } }{{\hfill $\Box$}\vspace{.5pc}}

\begin{document}
\begin{frontmatter}

\ifreport

\title{Stand Up Indulgent Gathering\footnote{This work was partially funded by the ANR project SAPPORO, ref. 2019-CE25-0005-1. Preliminary results appearing in this paper were presented at SSS 2020~\cite{BLT20c} and AlgoSensors 2021~\cite{BLT21c}.}}

\author[1]{Quentin Bramas}
\author[1]{Anissa Lamani}
\author[2]{S\'{e}bastien Tixeuil}

\affil[1]{University of Strasbourg, CNRS, ICUBE, France}
\affil[2]{Sorbonne University, CNRS, LIP6, France}

\maketitle

\fi

\begin{abstract}
We consider a swarm of mobile robots evolving in a bidimensional Euclidean space. We study a variant of the crash-tolerant gathering problem: if no robot crashes, robots have to meet at the same arbitrary location, not known beforehand, in finite time; if one or several robots crash at the same location, the remaining correct robots gather at the crash location to rescue them.
Motivated by impossibility results in the semi-synchronous setting, we present the first solution to the problem for the fully synchronous setting that operates in the vanilla Look-Compute-Move model with no additional hypotheses: robots are oblivious, disoriented, have no multiplicity detection capacity, and may start from arbitrary positions (including those with multiplicity points). We furthermore show that robots gather in a time that is proportional to the initial maximum distance between robots.
\end{abstract}


\end{frontmatter}

\section{Introduction}

The study of swarms of mobile robots from an algorithmic perspective was initiated by Suzuki and Yamashita~\cite{SuzukiY99}, and the simplicity of their model fostered many research results from the Distributed Computing community~\cite{2019Flocchini}. Such research focused in particular on the computational power of a set of autonomous robots evolving in a bidimensional Euclidean space. More precisely, the purpose is to determine which tasks can be achieved by the robots, and at what cost. 

Typically, robots are modeled as points in the two-dimensional plane, and are assumed to have their own coordinate system and unit distance. In addition, they are \emph{(i)} anonymous (they can not be distinguished), \emph{(ii)} uniform (they execute the same algorithm) and, \emph{(iii)} oblivious (they cannot remember their past actions).
Robots cannot interact directly but are endowed with visibility sensors allowing them to sense their environment and see the positions of the other robots. Robots operate in cycles that comprise three phases: Look, Compute and Move (LCM). During the first phase (Look), robots take a snapshot to see the positions of the other robots. During the second phase (Compute), they either compute a destination or decide to stay idle.  In the last phase (Move), they move towards the computed destination (if any).
Three execution models have been considered in the literature to represent the amount of synchrony between robots. The fully synchronous model (FSYNC) assumes all robots are activated simultaneously and execute their LCM cycle synchronously. The semi-synchronous model (SSYNC) allows that only a subset of robots is activated simultaneously. The asynchronous model (ASYNC) makes no assumption besides bounding the duration of each phase of the LCM cycle. 

Although a number of problems have been considered in this setting (\emph{e.g.}, pattern formation, scattering, gathering, exploration, patrolling, etc.), the gathering problem remains one of the most fundamental tasks one can ask mobile robots to perform, \emph{e.g.} to collect acquired data or to update the robots' algorithm. To solve gathering, robots are requested to reach the same geographic location, not known beforehand, in finite time. Despite its simplicity, the gathering problem yielded a number of impossibility results due to initial symmetry and/or adversarial scheduling. The special case of two robots gathering is called rendezvous in the literature.

When the number of robots in a swarm increases, the possibility of having one robot crash unpredictably becomes important to tackle. Previous work on fault-tolerant gathering considered two variants of the problem. In the weak variant, only correct (that is, non-failed) robots are requested to gather, while in the strong variant, all robots (failed and correct) must gather at the same point. Obviously, the strong variant is only feasible when all faulty robots crash at the same location. In this paper, we consider the strong version of the crash-tolerant gathering problem, which we call \emph{stand up indulgent gathering} (SUIG). In more detail, an algorithm solves the SUIG problem if the two following conditions are satisfied: \emph{(i)} if no robot crashes, all robots gather at the same location, not known beforehand, in finite time, and \emph{(ii)} if one or more robots crash at some location $l$, all robots gather at $l$ in finite time. The SUIG problem can be seen as a more stringent version of gathering, where correct robots are requested to rescue failed ones should any unforeseen event occur. Of course, robots are unaware of the crashed status of any robot as they have no means to communicate directly.

\paragraph{Related Works}
A fundamental result~\cite{SuzukiY99} shows that when robots operate in a fully synchronous manner, the rendezvous can be solved deterministically, while if some robots are allowed to wait for a while (this is the case \emph{e.g.} in the SSYNC model), the problem becomes impossible without additional assumptions. The impossibility result for SSYNC rendezvous extends to SSYNC gathering whenever bivalent initial configurations\footnote{A configuration is \emph{bivalent} if all robots are evenly spread on exactly two distinct locations.} are allowed~\cite{courtieu15ipl}, or when robots have no multiplicity detection capabilities~\cite{prencipe07tcs}. On the other hand, FSYNC gathering is feasible~\cite{cohen05siam}, even when robots have no access to multiplicity detection~\cite{balabonski19tcs}.  

Early solutions to weak crash-tolerant gathering in SSYNC for groups of at least three robots make use of extra hypotheses: \emph{(i)} starting from a distinct configuration (that is, a configuration where at most one robot occupies a particular position), at most one robot may crash~\cite{AgmonP06}, \emph{(ii)} robots are activated one at a time~\cite{DefagoGMP06}, \emph{(iii)} robots may exhibit probabilistic behavior~\cite{defago20dc}, \emph{(iv)} robots share a common chirality (that is, the same notion of handedness)~\cite{bouzid13icdcs}, \emph{(v)} robots agree on a common direction~\cite{BCM15}. It turns out that these hypotheses are \emph{not} necessary to solve deterministic weak crash-tolerant gathering in SSYNC, when up to $n-1$ robots may crash~\cite{bramas15wait}. 

The case of SUIG mostly yielded impossibility results: with at most a single crash, strong crash-tolerant gathering $n\geq3$ robots deterministically in SSYNC is impossible when robots do not have persistent coordinate systems, even if robots are executed one at a time~\cite{DefagoGMP06,defago20dc}, and probabilistic strong crash-tolerant gathering $n\geq3$ robots in SSYNC is impossible with a fair scheduler~\cite{DefagoGMP06,defago20dc}. However, probabilistic strong crash-tolerant gathering with $n\geq3$ robots becomes possible in SSYNC if the relative speed of the robots is upper bounded by a constant~\cite{DefagoGMP06,defago20dc}.

This raises the following open question: for $n\geq2$ robots, is it possible to deterministically gather all robots (if none of them crashes) at the same position, not known beforehand, or to deterministically gather at the position of the crashed robots, if some of them crash unexpectedly at the same location? Are additional hypotheses (such as multiplicity detection capability, or starting from a non-bivalent or even distinct configuration) necessary to solve the problem? \\

\paragraph{Contribution}

We answer positively to the feasibility of SUIG in mobile robotic swarms. Our approach is constructive, as we give an explicit algorithm to solve the SUIG problem in the FSYNC execution model, using only oblivious robots with arbitrary local coordinate systems (there is no agreement on any direction, chirality, or unit distance), without relying on multiplicity detection (robots are unaware of the number of robots, or approximation of it, on any occupied location). Furthermore, a nice property of our algorithm is that no assumptions are made about the initial positions of the robots (in particular, we allow to start from a configuration with multiplicity points, and even from a bivalent configuration), so the algorithm is self-stabilizing~\cite{dolev00book}.
The time complexity of our SUIG algorithm is $O\left(
\frac{D}{\delta} + 
\log\left(\frac{\delta}{u_{\mathit{min}}}\right) + 
\log^2\left(\frac{u_{\mathit{max}}}{u_{\mathit{min}}}\right)\right)$ rounds, where $u_{\mathit{max}}$, resp. $u_{\mathit{min}}$, denotes the maximum, resp. minimum, unit distance among robots, $D$ is the maximum distance between two robots in the initial configuration, and $\delta$ is the minimum distance a robot is assured to travel each round. 

To achieve this result, we introduce a new \emph{level-slicing} technique, that permits to partition robots that execute specific steps of the algorithm within a class of configurations, according to a level. This level permits to coordinate robots even when they do \emph{not} agree on a common unit distance. 
We believe that this technique could be useful for other problems related to mobile robotic swarms.

To illustrate the notion of levels, we first present a solution for the rendezvous problem in the presence of a single crash called SUIR problem (Stand Up Indulgent rendezvous). The rendezvous problem is a special case of the gathering when $k=2$. To justify the fully-synchronous setting assumed, we present an impossibility result that shows that even if robots have a common coordinate system, have full-lights\footnote{The \emph{full-lights} model corresponds to robots that are endowed with lights with various colors that can be changed by the hosting robot , and seen by other robots. Hence, robots can communicate explicitly with other robots.} with infinitely many colors, and their movements are rigid, no solution exists for the SUIR problem in SSYNC. 

\paragraph{Roadmap} We start by defining our model in Section~\ref{sec:model}. Next, we present our impossibility results in Section~\ref{sec:impossibility}. We warm up in Section~\ref{sec:rdv} by considering the Stand Up Indulgent Rendezvous (SUIR) problem, and then address the Stand Up Gathering (SUIG) problem in Section~\ref{sec:gethering}. Finally, we conclude and present some open questions in Section~\ref{sec:conclusion}. 

\section{Model}\label{sec:model}

Let $\mathcal{R}=\{r_1, r_2, \dots r_k\}$ be the set of $k \geq 2$ robots modeled as points in the Euclidean two-dimensional space. Robots are assumed to be anonymous (they are indistinguishable), uniform (they all execute the same algorithm), and oblivious (they cannot remember past actions). 

Let $Z$ be a global coordinate system. Let $p_i(t) \in \mathbb{R}^2$ be the coordinate of either a robot $r_i$ or a set of robots collocated with $r_i$ at time $t$ with respect to $Z$. Given a time $t$, a \emph{configuration} at time $t$ is defined as a set $\mathcal{C}_t=\{p_1(t), p_2(t), \dots p_m(t)\}$ ($m \leq k$) where  $m$ is the number of occupied points in the plane and $p_i(t) \ne p_j(t)$ for any $i,j \in \{1, \dots, m\}$, $i\neq j$. Robots do not know $Z$. Instead, each robot $r_i$ has its own coordinate system $Z_{r_i}$ centered at the current position of $r_i$. We assume \emph{disoriented} robots, \ie, they neither agree on an axis nor have a common unit distance. Observe that the Smallest Enclosing Circle (SEC) is independent of the coordinate system.

Robots operate in cycles of three phases: Look, Compute and Move. During the Look phase, a robot $r_i$ observes the positions of all the other robots and builds the set $V_{r_i}(t)  = \{p_1(t),\dots, p_m(t)\}$ ($m \leq k$) where  $m$ is the number of occupied points in the plane and $p_i(t)$ refers to the coordinate of a robot or a collocated set of robots with respect to $r_i$'s local coordinate system $Z_{r_i}$ (translated by $-r_i$ so that $r_i$ is always at the center). We call $V_{r_i}(t)$ the \emph{local view} of $r_i$ at time $t$. Note that robots have no way to distinguish positions occupied by a single robot from those occupied by more than one robot, \ie, they have no multiplicity detection. During the compute phase, $r_i$ computes a target destination with respect to its local coordinate system $Z_{r_i}$. Finally, during the move phase, $r_i$ moves towards the computed destination. We assume \emph{non-rigid} movements, \ie, when a robot moves towards a target destination computed during the Move phase, an adversary can prevent it from reaching it by stopping the robot anywhere along the straight path towards the destination. However, we assume that the robot travels at least a fixed positive distance $\delta$. The value of $\delta$ is common to all robots but is unknown. It can be arbitrarily small and does not change during the execution of the algorithm.

In the fully-synchronous model (FSYNC), all correct robots are activated at each time instant called a round and all robots execute their cycles synchronously. In the \emph{Semi-synchronous} model (SSYNC), only a non-empty subset of the correct robots may be activated at each round and similarly to the FSYNC model, the activated robots execute their cycle synchronously. In the latter case, we consider only \emph{fair} schedules, \ie, schedules where each correct robot is activated infinitely often.

\newcommand{\IT}{\mathcal{T}}

An algorithm $A$ is a function mapping local views to destinations. When $r_i$ is activated at time $t$, algorithm $A$ outputs $r_i$'s destination $d$ in its local coordinate system $Z_{r_i}$.  An execution $\mathcal{E} = (C_0, C_1, \dots)$ of $A$ is a sequence of configurations, where $C_0$ is an initial configuration, and every configuration $C_{t+1}$ is obtained from $C_{t}$ by applying $A$. 

A robot is said to \emph{crash} at time $t$ if it is not activated at any time $t' \geq t$, \ie, a crashed robot stops executing its algorithm and remains at the same position indefinitely.

\paragraph{The Stand Up Indulgent Gathering Problem}

An algorithm solves the Stand Up Indulgent Gathering (SUIG) problem if, for any initial configuration $C_0$ and for any execution $\mathcal{E}=(C_0, C_1, \ldots)$ with up to one crashed location (location where robots crash), there exists a round $t$ and a point $p$ such that $C_{t'} = \{p\}$ for all $t'\geq t$. In other words, if one robot crashes, all non-crashed robots go and join the crashed one; if no robot crashes, all robots gather in a finite number of rounds. 

Since we consider oblivious robots and an arbitrary initial configuration,  we can consider without loss of generality that the crash, if any, occurs at the start of the execution.

As mentioned previously, to illustrate the core idea of our proposed solution, let us first consider a special case of the SUIG problem called the Stand-Up Indulgent Rendezvous (SUIR), which addresses the gathering problem in the presence of at most one single crashed robot when $k = 2$.

\section{Impossibility Results}\label{sec:impossibility}

In this section, we first show that the SUIR problem is not solvable in SSYNC even if robots share a full coordinate system, and have access to infinite persistent memory that is readable by the other robot. In the literature~\cite{0001FPSY16}, the persistent memory aspect is called the Full-light model with an infinite number of colors. We then extend the result to show that without additional assumptions, the SUIG problem is also impossible to solve in SSYNC.

Let us first present a key point to prove our impossibility results.

\begin{lemma}\label{lem:execution contain move to other}
Consider the SSYNC model, with rigid movements, robots endowed with full-lights with infinitely many colors, and a common coordinate system. Assuming algorithm $A$ solves the SUIR problem, then, in every execution suffix starting from a distinct configuration where only robot $r$ is activated (\emph{e.g.} because the other robot has crashed), there must exist a configuration where algorithm $A$ commands that $r$ moves to the other robot's position.
\end{lemma}

\begin{proof}
Any move of robot $r$ that does not go to the other robot location does not yield gathering. If this repeats infinitely, no rendezvous is achieved.
\end{proof}

\begin{theorem}\label{thm:SSYNC Impossibility with lights}
The SUIR problem is not solvable in SSYNC, even with rigid movements, robots endowed with full-lights with infinitely many colors, and sharing a common coordinate system.
\end{theorem}

\begin{proof}
Assume for the purpose of contradiction that such an algorithm exists. Let $r$ be one of the robots, and $r'$ be the other robot. We construct a fair infinite execution where rendezvous is never achieved. At some round $t$, we either activate only $r$, only $r'$, or both, depending on what the (deterministic) output of the algorithm in the current configuration is:
\begin{itemize}
    \item If $r$ is dictated to stay idle: activate only $r$. In this case, only $r$ is activated in future rounds, the configuration never changes, and rendezvous is never achieved.
    \item If $r$ is dictated to move to $p\neq r'$: activate only $r$
    \item If $r$ is dictated to move to $r'$, and $r'$ is dictated to move: activate both robots.
    \item Otherwise ($r'$ is dictated to stay idle): activate only $r'$
\end{itemize}

We now show that the execution is fair. 
Suppose for the purpose of contradiction that the execution is unfair, so there exists a round $t$ after which only $r$ is executed, or only $r'$ is executed. In the first case, it implies there exists an execution suffix where $r$ is never dictated to move to the other robot, which contradicts Lemma~\ref{lem:execution contain move to other}. Now, if only $r'$ is activated after some round $t$, then there exists a suffix where $r'$ is always dictated to stay idle, which also contradicts Lemma~\ref{lem:execution contain move to other}.

The schedule we choose guarantees the following.
When only $r$ is activated, rendezvous is not achieved as $r$ is not moving to $r'$.
When only $r'$ is activated, rendezvous is not achieved as $r'$ is idle. If both robots are activated rendezvous is not achieved as $r$ is moving to $r'$ while $r'$ is moving. 

Overall, there exists an infinite fair execution where robots never meet, a contradiction with the initial assumption that the algorithm solves SUIR.
\end{proof}

\begin{corollary}\label{cor:impossibility of SUIG in SSYNC}
Without additional assumptions on the initial configuration, the SUIG problem is not solvable in SSYNC, even with rigid movements, robots endowed with full-lights with infinitely many colors, and sharing a common coordinate system.
\end{corollary}

\begin{proof}
Let us start from a bivalent configuration, \ie, a $2n$-robot configuration where $n$ robots are located at a point $p_1$ and $n$ robots are located at a point $p_2\neq p_1$. All the robots in $p_1$, resp. robots in $p_2$, can be considered as a single robot (by assuming they have the same coordinate system, same initial memory, and if the adversary always activates them together). Hence, the impossibility of the SUIR problem by Theorem~\ref{thm:SSYNC Impossibility with lights} implies the impossibility of the SUIG problem.
\end{proof}

\section{Stand-Up Indulgent Rendezvous}\label{sec:rdv}

We focus in this section on the Stand-Up Indulgent Rendezvous (SUIR) problem. The SUIR problem is a special case of the SUIG problem as it involves only two robots.  We present an algorithm to introduce the notion of levels, that solves the problem in FSYNC by oblivious disoriented robots.

We first prove that having an algorithm solving the SUIR problem in a one-dimensional space implies the existence of an algorithm solving the same problem in a two-dimensional space. This is important as our algorithms are defined in the one-dimensional space. However, observe that the converse is not true. Indeed, we also present an example of a problem (the fault-free rendezvous with one common full axis) that is solvable in a two-dimensional space, but that cannot be reduced to the one-dimensional space. Despite the results being intuitive, the formal proof is not trivial.

\subsection{Reduction To One-dimensional Space}

\begin{theorem}\label{thm:1D to 2D}
Suppose there exists an algorithm solving the stand up indulgent, resp. fault-free, rendezvous problem where robots are restricted to a one-dimensional space. Then, there exists an algorithm solving the stand up indulgent, resp. fault-free, problem in a two-dimensional space. 
\end{theorem}

\begin{proof}
Let $A_1$ be an algorithm solving the SUIR problem where robots are restricted to a one-dimensional space. We provide a constructive proof of the theorem by giving a new algorithm $A_2$ executed by robots in the two-dimensional space (the same procedure works with a fault-free rendezvous algorithm).

\newcommand{\vfunc}{\textbf{v}}

First, for a configuration $C = \{r_{min}, r_{max}\}$ of the robots in the two dimensional plane at time $t$, with $r_{min} < r_{max}$ (using the lexicographical order on their coordinates), we define the function $\vfunc$ as follows:

\[
\vfunc(\{r_{min}, r_{max}\}) = \frac{r_{\max} - r_{\min}}{\lVert r_{\max} - r_{\min} \rVert}
\]
Where the subtraction $b - a$ of two points $a$ and $b$ is the vector $\vec{ab}$ and $\lVert.\rVert$ denotes the Euclidean norm (so $\lVert r_{\max} - r_{\min} \rVert$ is the distance between $r_{\max}$ and $r_{\min}$). 

Let $r_1$ and $r_2$ denote the two robots, having transformation functions $h_1$ and $h_2$, respectively. Transformation functions are taken from a set of transformations containing rotations, scaling, reflections, and their compositions. We assume that the transformation function of a robot $r$ is chosen by an adversary, but does not change over time.
For each $i=1, 2$, let $V_i = \{(0,0), h_i(r_{3-i} - r_i) \}$ be the local view of robot $r_i$ at the current time $t$ (\ie, the set of positions of the robots seen in the local coordinate system of robot $r_i$). Each robot $r_i$ can compute its own orientation vector $v_i = \vfunc(V_i)$ of the line joining the two robots. Notice that, if robots remain on the same line, then $v_i$ remains invariant during the whole execution as long as robots do not gather. The top-left part of Figure~\ref{fig:mapping view 2D -> 1D} illustrates the view of robot $r_1$ and its vector $v_1$.

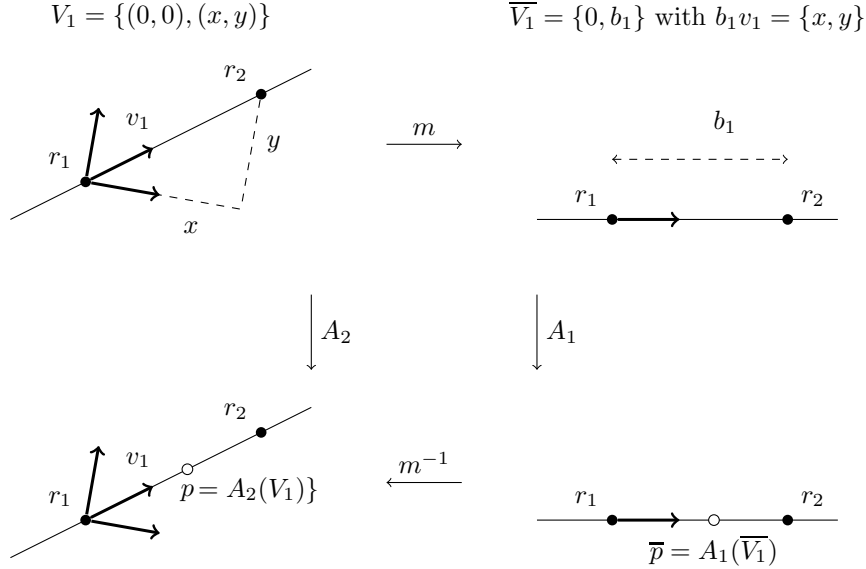
\begin{figure}
    \centering
    \begin{tikzpicture}
    \draw (0,0) -- (4,2);
    
    \node[fill=black, circle, inner sep=0.05cm, label={135:$r_1$}] (r1) at (4/4, 2/4) {};
    \node[fill=black, circle, inner sep=0.05cm, label={135:$r_2$}] (r2) at (4/1.2, 2/1.2) {};
    
    \draw[->, line width=0.04cm] (r1) -- ++ (80:1);
    \draw[->, line width=0.04cm] (r1) -- ++ (-10:1);
    \draw[->, line width=0.04cm] (r1) -- ++ ({2/(sqrt(5)},{1/(sqrt(5)});

    \node at (2, 2.7) {$V_1 = \{(0,0), (x,y)\}$};
    \draw[dashed] (r1) -- ++ (-10:2.1);
    \draw[dashed] (r2) -- ++ (-100:1.55);
    \node at (2.4, -0.1) {$x$};
    \node at (3.5, 1) {$y$};
    \node at (1.7, 1.3) {$v_1$};
    
    \draw[->] (4,-1) -- node[right] {$A_2$} (4,-2);
    \draw[->] (5,1) -- node[above] {$m$} (6,1);
    
    \begin{scope}[xshift=7cm, yshift=0cm]
    \node[fill=black, circle, inner sep=0.05cm, label={135:$r_1$}] (r1) at (4/4, 0/4) {};
    \draw (0,0) -- (4,0);
    \draw[->, line width=0.04cm] (r1) -- ++ ({2/(sqrt(5)},0);
    
    \node[fill=black, circle, inner sep=0.05cm, label={45:$r_2$}] (r2) at (4/1.2, 0) {};
    
    \node at (2, 2.7) {$\overline{V_1} = \{0, b_1\}$ with $b_1v_1 = \{x,y\}$};
    
    \draw[<->, dashed] (1,0.8) -- (4/1.2,0.8);
    \node at (2.5, 1.3) {$b_1$};
    
    \draw[->] (0,-1) -- node[right] {$A_1$} (0,-2);
    \end{scope}

    \begin{scope}[xshift=7cm, yshift=-4cm]
    \node[fill=black, circle, inner sep=0.05cm, label={135:$r_1$}] (r1) at (4/4, 0/4) {};
    \draw (0,0) -- (4,0);
    \draw[->, line width=0.04cm] (r1) -- ++ ({2/(sqrt(5)},0);
    
    \node[fill=black, circle, inner sep=0.05cm, label={45:$r_2$}] (r2) at (4/1.2, 0) {};
    \node[fill=white, draw=black, circle, inner sep=0.05cm, label={-90:$\overline{p}=A_1(\overline{V_1})$}] at (4/1.7, 0) {};

\end{scope}

\begin{scope}[yshift=-4.5cm]
    \draw (0,0) -- (4,2);
    
    \node[fill=black, circle, inner sep=0.05cm, label={135:$r_1$}] (r1) at (4/4, 2/4) {};
    \node[fill=black, circle, inner sep=0.05cm, label={135:$r_2$}] (r2) at (4/1.2, 2/1.2) {};

    \draw[->, line width=0.04cm] (r1) -- ++ (80:1);
    \draw[->, line width=0.04cm] (r1) -- ++ (-10:1);
    \draw[->, line width=0.04cm] (r1) -- ++ ({2/(sqrt(5)},{1/(sqrt(5)});
    \node at (3.3, 0.9) {$=A_2(V_1)\}$};
    
    \node at (1.7, 1.3) {$v_1$};

    \node[fill=white, draw=black, circle, inner sep=0.05cm, label={-90:$p$}] at (4/1.7, 2/1.7) {};
    
    \draw[<-] (5,1) -- node[above] {$m^{-1}$} (6,1);
\end{scope}    
    
    \end{tikzpicture}
    \caption{Relation between the views of $r_1$ in the two-dimensional plane, and in the one-dimensional line. The target destination in the plane is obtained by applying the inverse mapping of the target destination computed on the line.}
    \label{fig:mapping view 2D -> 1D}
\end{figure}

We define algorithm $A_2$, executed by robots in the two-dimensional space as follows. First, if the local view of a robot $r_i$ is $\{(0,0)\}$, then $A_2$ outputs $(0,0)$. 

Otherwise, $r_i$ can map its local view $V_i$ in a one-dimensional space to obtain $\overline{V_i} = \{0, b_i\}$ with $b_i$ such that $h_i(r_{3-i} - r_i) = b_i v_i$, and execute $A_1$ on $\overline{V_i}$. The obtained destination $\overline{p} = A_1(\overline{V_i})$, is then converted back to the two-dimensional space to obtain the destination $p = \overline{p} v_i$, as illustrated in Figure~\ref{fig:mapping view 2D -> 1D}.
In doing so, the robots remain on the same line, and $v_i$ remains invariant while the robots are not gathered (when they are gathered, both algorithms stop).

Let $E = \left(C_0, C_1, \ldots \right)$ be an arbitrary execution of $A_2$. We want to construct from $E$ an execution $\overline{E}$ of $A_1$ such that the rendezvous is achieved in $\overline{E}$ if and only if the rendezvous is achieved in $E$.

Recall that we analyze each configuration $C_i$, $i\in\N$, using $Z$, the global coordinate system we use for the analysis. Let $v = \vfunc(C_0)$. Again, since robots remain on the same line, then $v = \vfunc(C_i)$ for any $C_i$ while robots are not yet gathered.

Let $O$ be any point of the line $L$ joining the two robots.
We define as follows the bijection $m$ mapping points of $L$ (in $Z$), to the global one-dimensional coordinate system $(O,v)$:
\[
\forall a\in\R, m(O + a v) = a
\]
We can extend the function $m$ to configurations as follows: \\$m(C) = \{m(r)\;|\; r\in C\}$.

Now, let $\overline{E}=\left(\overline{C_0}, \overline{C_1},\ldots\right)$ be the execution of $A_1$, in $(O,v)$, of two robots having transformation function $\overline{h}_1$ and $\overline{h}_2$ respectively, with $\overline{h}_i(a) = b$ if and only if $h_i(a v) = b v_i$.

\tolerance=1000
We now show that, if $C\overset{A_2}{\rightarrow} C'$ then $m(C)\overset{A_1}{\rightarrow} m(C')$. To do so we show that the result of executing $A_1$ on $m(C)$ coincides with $m(C')$. 
Let ${C=\{O+a_1v, O+a_2v\}}$, $i$ be an activated robot and $j$ be the other robot. On one hand, we have $m(C) = \{a_1, a_2\}$ and, by construction, the view $\overline{V_i}$ of robot $i$ in $m(C)$ is $\{0, b_i\}$ with $b_i = \overline{h}_i(a_j - a_i)$, so that the global destination of $r_i$ in $m(C)$ is then $a_i + \overline{h}_i^{-1}(\overline{p})$ (with $\overline{p} = A_1(\overline{V_i})$). 
On the other hand, the view ${V_i}$ of robot $i$ is $\{(0,0), h_i((a_j-a_i)v)\} = \{(0,0), b_iv_i\}$, so that the global destination of $r_i$ in $Z$ is $O+a_iv + h_i^{-1}(\overline{p}v_i) = O+a_iv + \overline{h}_i^{-1}(\overline{p})v $. Since $m( O+a_iv + \overline{h}_i^{-1}(\overline{p})v) = a_i + \overline{h}_i^{-1}(\overline{p})$, we obtain that $m(C)\overset{A_1}{\rightarrow} m(C')$ (assuming the same robots are activated in $C$ and in $m(C)$.

Hence, $\overline{C_i} = m(C_i)$ for all $i\in \N$.
Since $A_1$ solves the SUIR problem, there exists a point $p\in \R$ and a round $t$ such that, for all $t'\geq t$, $m(C_{t'}) = \{p\}$. This implies that $C_{t'} = \{ O + p v \}$, so that $A_2$ solves the SUIR problem.
\end{proof}

Observe that the converse of Theorem~\ref{thm:1D to 2D} is not true. To prove this statement, we present an algorithm solving the (fault-free) rendezvous problem in a two-dimensional space, assuming robots agree on one full axis (that is, they agree on the direction and the orientation of the axis). Under this assumption, it is possible that the robots do not agree on the orientation of the line joining them, so that assuming the converse of Theorem~\ref{thm:1D to 2D} would imply the existence of an algorithm in the one-dimensional space with disoriented robots (which does not exist, using a similar proof as the one given in~\cite{SuzukiY99}).

The idea is that, if the configuration is symmetric (robots may have the same view), then robots move to the point that forms, with the two robots, an equilateral triangle. Since two such points exist, the robots choose the northernmost one (the robots agree on the $y$-axis, which provides a common North). Otherwise, the configuration is not symmetric, and there is a unique northernmost robot $r$. This robot does not move and the other robot moves to $r$.

\begin{algorithm}[H] 
\KwData{
  $r$ : robot executing the algorithm
 }
Let $\{(0,0), (x,y)\}$ be $r$'s local view.\\
\uIf{$y = 0$}{
    \nl\label{algo:fault-free rdv:move to triangle} move to the point $\left(x/2, \left|\frac{x\sqrt{3}}{2}\right|\right)$
}
\uElseIf{$y > 0$}{
    \nl\label{algo:fault-free rdv:move to other} move to the other robot, at $(x,y)$.
}
 \caption{Fault-free rendezvous. Robots agree on one full axis, may not have a common unit distance, and movements are non-rigid}\label{algo:one common axis}
\end{algorithm}

\begin{theorem}\label{thm: proof of algo, one axis, non-rigid}
Algorithm~\ref{algo:one common axis} solves the (fault-free) rendezvous problem with non-rigid movements and robots having only one common full axis, and different unit distances.
\end{theorem}

\begin{proof}
If the configuration is not symmetric, then, after executing Line~\ref{algo:fault-free rdv:move to other}, either the rendezvous is achieved, or the moving robots remain on the same line joining the robots, so the obtained configuration is still asymmetric, and the same robot is dictated to move towards the same destination, so it reaches it in a finite number of rounds.

Consider now that the initial configuration $C$ is symmetric. 
If both robots reach their destination, the rendezvous is completed in one round. If robots are stopped before reaching their destinations, two cases can occur. Either they are stopped after traveling different distances, or they are stopped at the same $y$-coordinate. In the former case, the obtained configuration is asymmetric and we retrieve the first case. In the latter, the configuration remains symmetric, but the distance between the two robots decreases by at least $\frac{2\delta}{\sqrt{2}}$.
Hence, at each round either robots complete the rendezvous, reach an asymmetric configuration, or come closer by a fixed distance. Since the latter case cannot occur infinitely, one of the other cases occurs at least once, and the rendezvous is completed in a finite number of rounds.
\end{proof}

\subsection{SUIR Algorithm.}

We consider in the following disoriented robots (\ie they agree neither on an axis, nor on a common unit distance) and assume that the move operations are rigid. We present Algorithm~\ref{algo:disoriented FSYNC} to solve the SUIR problem under such assumptions. The algorithm is defined on the line. Each robot sees the line oriented in some way, but robots might not agree on the orientation of the line. However, since the orientation of the line is deduced from the robot's own coordinate system, it does not change over time.
The different moves of a robot $r$ depend on whether $r$ sees itself on the left or the right of the other robot, and on its \emph{level}. The level of a robot at distance $d$ from the other robot (according to its own coordinate system, hence its own unit distance) is the integer $i\in\Z$ such that $d\in [2^{-i}, 2^{1-i})$.
In the pseudo-code, we write

\[
\mathit{left}\rightarrow \texttt{Move1} ; \mathit{right}\rightarrow \texttt{Move2}
\] 
to notify that a robot that sees itself on the left, resp. on the right, should execute \texttt{Move1}, resp. \texttt{Move2}.

\begin{algorithm}[H]
Let $d$ be the distance to the other robot

Let $i\in\mathbb{Z}$ such that $d\in [2^{-i}, 2^{1-i})$\;
\lIf{$i \equiv 0 \mod 2$}{
    move to the middle
}
\lIf{$i \equiv 1 \mod 4$}{
    \emph{left} $\rightarrow$ move to middle ;
    \emph{right} $\rightarrow$ move to other
}
\lIf{$i \equiv 3 \mod 4$}{
    \emph{left} $\rightarrow$ move to other ;
    \emph{right} $\rightarrow$ move to middle
}
 \caption{SUIR Algorithm for disoriented robots}\label{algo:disoriented FSYNC}
\end{algorithm}

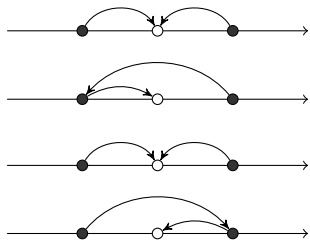
\begin{figure}[h]
    \centering
\begin{tikzpicture}
\draw[->] (-1,0) -- (3,0);
\node[robot] (a) at (0,0) {};
\node[robot] (b) at (2,0) {};
\node[dest] (m) at (1,0) {};
\node[opacity=0, text opacity=1] at (-4,0) {case $i\equiv 0\mod 4$};

\path (a) edge[bend left=60,->,>=stealth'] node [left] {} (m);
\path (b) edge[bend right=60,->,>=stealth'] node [left] {} (m);
\end{tikzpicture}

\begin{tikzpicture}
\draw[-, white] (-5,-0.2) -- (2.5,0.5); \draw[->] (-1,0) -- (3,0);
\node[robot] (a) at (0,0) {};
\node[robot] (b) at (2,0) {};
\node[dest] (m) at (1,0) {};
\node[opacity=0, text opacity=1] at (-4,0) {case $i\equiv 1\mod 4$};

\path (a) edge[bend left=30,->,>=stealth'] node [left] {} (m);
\path (b) edge[bend right=50,->,>=stealth'] node [left] {} (a);
\end{tikzpicture}

\begin{tikzpicture}
\draw[-, white] (-5,-0.2) -- (2.7,0.6); \draw[->] (-1,0) -- (3,0);
\node[robot] (a) at (0,0) {};
\node[robot] (b) at (2,0) {};
\node[dest] (m) at (1,0) {};
\node[opacity=0, text opacity=1] at (-4,0) {case $i\equiv 2\mod 4$};

\path (a) edge[bend left=60,->,>=stealth'] node [left] {} (m);
\path (b) edge[bend right=60,->,>=stealth'] node [left] {} (m);
\end{tikzpicture}

\begin{tikzpicture}
\draw[-, white] (-5,-0.2) -- (3,0.5); \draw[->] (-1,0) -- (3,0);
\node[robot] (a) at (0,0) {};
\node[robot] (b) at (2,0) {};
\node[dest] (m) at (1,0) {};
\node[opacity=0, text opacity=1] at (-4,0) {case $i\equiv 3\mod 4$};

\path (a) edge[bend left=50,->,>=stealth'] node [left] {} (b);
\path (b) edge[bend right=30,->,>=stealth'] node [left] {} (m);
\end{tikzpicture}
    \caption{The four possible configurations, depending the distance between the two robots. We have split the case $i\equiv 0 \mod 2$ into two lines to help the reader.}
    \label{fig:four possible configurations}
\end{figure}

Figure~\ref{fig:four possible configurations} summarizes the eight possible views of a robot $r$, and the corresponding movements. Each line represents the congruence of the level of the robot modulo four, and on each line, we see the movement of the robot whether it sees itself on the right or on the left of the other robot. Note that the cases described in Figure~\ref{fig:four possible configurations} do not necessarily imply that both robots actually perform the corresponding movement at the same time (since they may have a different view).

For instance, if a robot $r_1$ has a level $i_1$ congruent to 1 modulo 4 and sees itself on the right, while the other robot $r_2$ has a level $i_2$ congruent to 2 modulo 4, and also sees itself on the right, then $r_1$ moves to the other robot position, and $r_2$ moves to the middle. Assuming both robots reach their destination, then the distance between them is divided by two (regardless of the coordinate system) so their levels increase by one, and they both see the other robot on the other side, so each robot now sees the other robot on its left. Now, $i_1$ congruent to 2 modulo 4 while $i_2$ is congruent to 3 modulo 4. As both robots observe themselves on the left side, $r_1$ moves to the middle while $r_2$ moves to the other robot's position. Again, as the distance between them is divided by two, their levels increase by one and since they have switched their positions they again observe themselves on the right side. In the new configuration, $i_1$ congruent to 3 modulo 4 while $i_2$ is congruent to 0 modulo 4. By executing Algorithm \ref{algo:disoriented FSYNC}, both robots move to the middle and the rendezvous is achieved. Figure~\ref{fig:exple} illustrates this execution. 

\begin{figure}[h]
    \centering
\begin{tikzpicture}
\draw[-, white] (-6,-0.2) -- (2,0.5); \draw[-] (-2,0) -- (2,0);
\node[robot] (a) at (-1,0) {};
\node[robot] (b) at (1,0) {};
\node[dest] (m) at (0,0) {};

\node[opacity=0, text opacity=1] at (-4,0) {(right) $i_1\equiv 1\mod 4$};
\node[opacity=0, text opacity=1] at (4,0) {(right) $i_2\equiv 2\mod 4$};

\path (a) edge[bend left=50,->,>=stealth'] node [left] {} (b);
\path (b) edge[bend right=30,->,>=stealth'] node [left] {} (m);

\end{tikzpicture}

\begin{tikzpicture}
\draw[-, white] (-6,-0.2) -- (2,0.5); \draw[-] (-2,0) -- (2,0);
\node[robot] (a) at (-1,0) {};
\node[robot] (b) at (1,0) {};
\node[dest] (m) at (0,0) {};

\node[opacity=0, text opacity=1] at (-4,0) {(left) $i_2\equiv 3\mod 4$};
\node[opacity=0, text opacity=1] at (4,0) {(left) $i_1\equiv 2\mod 4$};

\path (a) edge[bend left=50,->,>=stealth'] node [left] {} (b);
\path (b) edge[bend right=30,->,>=stealth'] node [left] {} (m);
\end{tikzpicture}

\begin{tikzpicture}
\draw[-, white] (-6,-0.2) -- (2,0.5); \draw[-] (-2,0) -- (2,0);
\node[robot] (a) at (-1,0) {};
\node[robot] (b) at (1,0) {};
\node[dest] (m) at (0,0) {};

\node[opacity=0, text opacity=1] at (-4,0) {(right) $i_1\equiv 3\mod 4$};
\node[opacity=0, text opacity=1] at (4,0) {(right) $i_2\equiv 0\mod 4$};

\path (a) edge[bend left=50,->,>=stealth'] node [left] {} (m);
\path (b) edge[bend right=50,->,>=stealth'] node [left] {} (m);

\end{tikzpicture}
    \caption{A possible execution of Algorithm \ref{algo:disoriented FSYNC} }
    \label{fig:exple}
\end{figure}
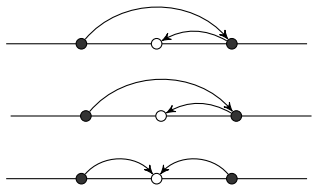

\subsection{Proof of Correctness.} We show in the following the correctness of Algorithm \ref{algo:disoriented FSYNC}. Let $C$ be any configuration and $d$ is the distance (in the global coordinate system $Z$) between the two robots. Let $x$, resp. $y$, be the distance, in $Z$, traveled by the left robot, resp. the right robot. Since the robots move toward each other, after executing one round, the distance between the robot becomes $f(d,x,y) = |d - x - y|$.

\begin{lemma}\label{lem:f decreases by fixed distance}
If at least one robot is dictated to move to the middle, then we have $f(d,x,y) \leq d - \min(\delta, d/2)$.
\end{lemma}

\begin{proof}
For any fixed $d$, using the symmetry of $f$ (with respect to the second and third argument), we have $f(d,x,y) = g_d(x+y)$ with $g_d: w\mapsto |d - w|$.
We know that the distance traveled by the robots is either $0$ (if one robot crashes), or at least $\min(d/2, \delta)$, but we cannot have $x=y=0$. Also, since at least one robot moves to the middle, we have either \emph{(i)} $x \leq d/2$ and $y\leq d$, or \emph{(ii)} $x \leq d$ and $y\leq d/2$. Hence, the sum $x+y$ is in the interval $[\min(d/2, \delta), 3d/2]$. 

As a convex function, the maximum of $g_d$ is reached at the boundary of its domain
\begin{align*}
    f(d,x,y) = g_d(x+y) &\leq \max(g_d(3d/2), g_d(\min(d/2, \delta)))\\
&\leq \max(d/2, d - \min(d/2, \delta)) =  d - \min(d/2, \delta)
\end{align*}
\end{proof}

The next lemma is a direct consequence of Lemma~\ref{lem:f decreases by fixed distance}.{
\begin{lemma}\label{lem:distance decreases by fixed distance}
Two rounds after any configuration where robots are at distance $d$ (in $Z$), the distance between the robots decreases by at least $\min(\delta, d/2)$.
\end{lemma}
}{
\begin{proof}
From Lemma \ref{lem:f decreases by fixed distance}, we know that if at least one robot is dictated to move to the middle, then, after one round, the distance between the two robots decreases by at least $\min(\delta, d/2)$. 
Otherwise,  we know that both robots were dictated to move to the other robot's location. Hence, the distance between the robots is either at most $d - \min(\delta, d/2)$, or is greater than $d - \min(\delta, d/2)$ (but still at most $d$). In the former case, the lemma is proved. In the latter case the order of the robots changes (a left robot becomes a right robot, and vice versa). This happens regardless of their coordinates system (maybe both robots view themselves on the right, then they both view themselves on the left). Also, the level of each robot is either the same, or is increased by one. In all cases, both robots are dictated to move to the middle.
In more detail, a right robot with level $i\equiv 1\mod 4$ becomes a left robot with either the same level of level $i+1\equiv 2\mod 4$. In both cases, in the next round, the robot is dictated to move to the middle.
A left robot with level $i\equiv 3\mod 4$ becomes a right robot with either the same level of level $i+1\equiv 0\mod 4$. In both cases, in the next round, the robot is dictated to move to the middle.

So that after one more round, the distance decreases by at least $\min(\delta, d/2)$.
\end{proof}
}

From the previous Lemma, we know there exists a round after which the robots are (and remain) at distance at most $\delta$, so it is enough to prove the correctness of our algorithm assuming movements are rigid (robots always reach their target destination).

If, at round $t$, one robot sees itself on the right, and the other sees itself on the left, then they agree on the orientation of the line at time $t$. Since, for each robot, the orientation of the line does not change, then they agree on it during the whole execution (except when they are gathered, as the line is not defined in that case).

Similarly, if at some round, both robots see themselves at the right (resp. at the left), then their orientations of the line are opposite, and remain opposite during the whole execution (again, until they gather). Hence we have the following remark.

\begin{remark}\label{rem:orientation common or opposite}
Consider two disoriented robots moving on the line $L$ joining them and executing Algorithm~\ref{algo:disoriented FSYNC}. Then, either they have a common orientation of $L$ during the whole execution (while they are not gathered), or they have opposite orientations of $L$ during the whole execution (while they are not gathered).
\end{remark}

\begin{lemma}\label{lem: proof of algo, correct, common}
Assuming rigid movements, no crash, and robots having \textbf{a common orientation} of the line joining them, then, Algorithm~\ref{algo:disoriented FSYNC} solves the SUIR problem.
\end{lemma}

\begin{proof}
Since the robots have a common orientation, we know there is one robot that sees itself on the right and one robot that sees itself on the left. Of course, the robots are not aware of this, but we saw in the previous remark that a common orientation is preserved during the whole execution (while robots are not gathered).

Let $(i,j) \in \Z^2$ denote a configuration where the robot on the left is at level $i$, and the robot on the right is at level $j$, and we write $(i,j)\equiv (k,l) \mod 4$ if and only if $i\equiv k \mod 4$ and $j\equiv l \mod 4$.

To prove the lemma we want to show that for any configuration $(i,j) \in \Z^2$, the robots achieve rendezvous. 
Take an arbitrary configuration $(i,j) \in \Z^2$. We consider all 16 cases:

\begin{enumerate}
\item {\small \label{i:com:0-0} \textbf{if} $(i,j)\equiv (0,0) \mod 4$: rendezvous is achieved in one round.}
\item {\small\label{i:com:0-1} \textbf{if} $(i,j)\equiv (0,1) \mod 4$:  we reach configuration $(j+1, i+1) \equiv (2,1) \mod 4$}
\item {\small\label{i:com:0-2} \textbf{if} $(i,j)\equiv (0,2) \mod 4$: rendezvous is achieved in one round.}
\item {\small\label{i:com:0-3} \textbf{if} $(i,j)\equiv (0,3) \mod 4$: rendezvous is achieved in one round.}
\item {\small\label{i:com:1-0} \textbf{if} $(i,j)\equiv (1,0) \mod 4$: rendezvous is achieved in one round.}
\item {\small\label{i:com:1-1} \textbf{if} $(i,j)\equiv (1,1) \mod 4$:  we reach configuration $(j+1, i+1) \equiv (2,2) \mod 4$}
\item {\small\label{i:com:1-2} \textbf{if} $(i,j)\equiv (1,2) \mod 4$: rendezvous is achieved in one round.}
\item {\small\label{i:com:1-3} \textbf{if} $(i,j)\equiv (1,3) \mod 4$: rendezvous is achieved in one round.}
\item {\small\label{i:com:2-0} \textbf{if} $(i,j)\equiv (2,0) \mod 4$: rendezvous is achieved in one round.}
\item {\small\label{i:com:2-1} \textbf{if} $(i,j)\equiv (2,1) \mod 4$:  we reach configuration $(j+1, i+1) \equiv (2,3) \mod 4$}
\item {\small\label{i:com:2-2} \textbf{if} $(i,j)\equiv (2,2) \mod 4$: rendezvous is achieved in one round.}
\item {\small\label{i:com:2-3} \textbf{if} $(i,j)\equiv (2,3) \mod 4$: rendezvous is achieved in one round.}
\item {\small\label{i:com:3-0} \textbf{if} $(i,j)\equiv (3,0) \mod 4$:  we reach configuration $(j+1, i+1) \equiv (1,0) \mod 4$}
\item {\small\label{i:com:3-1} \textbf{if} $(i,j)\equiv (3,1) \mod 4$:  we reach configuration $(j, i) \equiv (1,3) \mod 4$}
\item {\small\label{i:com:3-2} \textbf{if} $(i,j)\equiv (3,2) \mod 4$:  we reach configuration $(j+1, i+1) \equiv (3,0) \mod 4$}
\item {\small\label{i:com:3-3} \textbf{if} $(i,j)\equiv (3,3) \mod 4$:  we reach configuration $(j+1, i+1) \equiv (0,0) \mod 4$}
\end{enumerate}

In any case, rendezvous is achieved after at most three rounds.
\end{proof}

\begin{lemma}\label{lem: proof of algo, correct, opposite}
Assuming rigid movement, no crash, and robots having \textbf{opposite orientations} of the line joining them, then, Algorithm~\ref{algo:disoriented FSYNC} solves the SUIR problem.
\end{lemma}

\begin{proof}
Since the robots have opposite orientations, we know they either both see themselves on the right or they both see themselves on the left. Of course, the robots are not aware of this, but we saw in the previous remark that the opposite orientations are preserved during the whole execution (while robots are not gathered).

In this proof, $R\{i,j\}$ denotes a configuration where both robots see themselves on the right and one of them has level $i$ and the other as level $j$. Here, the order between $i$ and $j$ does not matter (hence the set notation). Similarly $L\{i,j\}$ denotes a configuration where both robots see themselves on the left, and one of them has level $i$ and the other as level $j$. 

Here, assuming without loss of generality that $i\leq j \mod 4$, we write $R\{i,j\}\equiv (k,l) \mod 4$, resp. $L\{i,j\}\equiv (k,l) \mod 4$, if and only if, $i\equiv k \mod 4$ and $j\equiv l \mod 4$.

To prove the lemma, we want to show that for any configuration $R\{i,j\}$ or $L\{i,j\}$, the robots achieve rendezvous.
Take an arbitrary configuration $(i,j) \in \Z^2$. We consider all 20 cases:

\begin{enumerate}
\item {\small \label{i:opp:L0-0} \textbf{if} $L\{i,j\}\equiv (0,0) \mod 4$: rendezvous is achieved in one round.}
\item {\small\label{i:opp:L0-1} \textbf{if} $L\{i,j\}\equiv (0,1) \mod 4$: rendezvous is achieved in one round.}
\item {\small\label{i:opp:L0-2} \textbf{if} $L\{i,j\}\equiv (0,2) \mod 4$: rendezvous is achieved in one round.}
\item {\small \label{i:opp:L0-3} \textbf{if} $L\{i,j\}\equiv (0,3) \mod 4$: we reach configuration $R\{i{+}1, j{+}1\}{\equiv} (0,1) \mod 4$.}

\item {\small\label{i:opp:L1-1} \textbf{if} $L\{i,j\}\equiv (1,1) \mod 4$: rendezvous is achieved in one round.}
\item{\small \label{i:opp:L1-2} \textbf{if} $L\{i,j\}\equiv (1,2) \mod 4$: rendezvous is achieved in one round.}
\item{\small \label{i:opp:L1-3} \textbf{if} $L\{i,j\}\equiv (1,3) \mod 4$: we reach configuration $R\{i{+}1, j{+}1\}{\equiv} (0,2) \mod 4$.}
\item{\small \label{i:opp:L2-2} \textbf{if} $L\{i,j\}\equiv (2,2) \mod 4$: rendezvous is achieved in one round.}
\item{\small \label{i:opp:L2-3} \textbf{if} $L\{i,j\}\equiv (2,3) \mod 4$: we reach configuration $R\{i{+}1, j{+}1\}{\equiv} (0,3) \mod 4$.}
\item{\small \label{i:opp:L3-3} \textbf{if} $L\{i,j\}\equiv (3,3) \mod 4$: we reach configuration $R\{i, j\}{\equiv} (3,3) \mod 4$.}

\item {\small\label{i:opp:R0-0} \textbf{if} $R\{i,j\}\equiv (0,0) \mod 4$: rendezvous is achieved in one round.}
\item {\small \label{i:opp:R0-1} \textbf{if} $R\{i,j\}\equiv (0,1) \mod 4$: we reach configuration $L\{i{+}1, j{+}1\}{\equiv} (1,2) \mod 4$.}
\item {\small\label{i:opp:R0-2} \textbf{if} $R\{i,j\}\equiv (0,2) \mod 4$: rendezvous is achieved in one round.}
\item{\small \label{i:opp:R0-3} \textbf{if} $R\{i,j\}\equiv (0,3) \mod 4$: rendezvous is achieved in one round.}
\item {\small\label{i:opp:R1-1} \textbf{if} $R\{i,j\}\equiv (1,1) \mod 4$: we reach configuration $L\{i, j\}{\equiv} (1,1) \mod 4$.}
\item {\small\label{i:opp:R1-2} \textbf{if} $R\{i,j\}\equiv (1,2) \mod 4$: we reach configuration $L\{i{+}1, j{+}1\}{\equiv} (2,3) \mod 4$.}
\item {\small\label{i:opp:R1-3} \textbf{if} $R\{i,j\}\equiv (1,3) \mod 4$: we reach configuration $L\{i{+}1, j{+}1\}{\equiv} (0,2) \mod 4$.}
\item {\small\label{i:opp:R2-2} \textbf{if} $R\{i,j\}\equiv (2,2) \mod 4$: rendezvous is achieved in one round.}
\item {\small\label{i:opp:R2-3} \textbf{if} $R\{i,j\}\equiv (2,3) \mod 4$: rendezvous is achieved in one round.}
\item {\small\label{i:opp:R3-3} \textbf{if} $R\{i,j\}\equiv (3,3) \mod 4$: rendezvous is achieved in one round.}
\end{enumerate}

In any case, rendezvous is achieved after at most three rounds.
\end{proof}

\begin{lemma}\label{lem: proof of algo, rigid, crash}
Assuming rigid movement and one robot crash, Algorithm~\ref{algo:disoriented FSYNC} solves the SUIR problem.
\end{lemma}

\begin{proof}
Let $i$ be the level of the correct robot $r$. Assume the other robot crashes. Robot $r$ either sees itself on the right or on the left of the other robot.

If $r$ sees itself on the right, then depending on its level, either $r$ moves to the middle, or moves to the other robot. In the former case, the level of $r$ increases by one and $r$ continues to see itself on the right. In the latter case, the rendezvous is achieved in one round. After at most three rounds, the level of $r$ is congruent to 1 modulo 4 so that after at most four rounds the rendezvous is achieved.

Similarly, if $r$ sees itself on the left, then after at most four rounds, $r$'s level is congruent to 3 modulo 4 and the rendezvous is achieved.
\end{proof}

\begin{theorem}
Algorithm~\ref{algo:disoriented FSYNC} solves the SUIR problem with disoriented robots in FSYNC.
\end{theorem}

\begin{proof}
By Lemma~\ref{lem:distance decreases by fixed distance}, the distance between the two robots decreases by at least $\min(\delta, d/2)$ every two rounds. Hence, eventually, robots are at distance smaller than $\delta$ from one another and, from this point in time, movements are rigid. Assume now that movements are rigid.
If a robot crashes, then the rendezvous is achieved by using Lemma~\ref{lem: proof of algo, rigid, crash}. Otherwise, depending on whether the robots have a common orientation or opposite orientation of the line joining them (see Remark~\ref{rem:orientation common or opposite}), the Theorem follows either by using Lemma~\ref{lem: proof of algo, correct, common} or by using Lemma~\ref{lem: proof of algo, correct, opposite}.
\end{proof}

\section{The Stand-Up Indulgent Gathering}\label{sec:gethering}

Now that we introduced the notion of levels through the case of the SUIR problem, let us focus on the stand up indulgent gathering problem where $k \geq 2$. One might think that it is sufficient to make the good robots gather in one location and then apply Algorithm \ref{algo:disoriented FSYNC} to perform the gathering. But this is not sufficient. Indeed, as the good robots might have different levels and orientations, by executing Algorithm \ref{algo:disoriented FSYNC}, they might be split into several locations. Hence, Algorithm \ref{algo:disoriented FSYNC} cannot be applied directly.  

\label{sec:algo}

\newcommand{\move}[1]{\texttt{Move}\ensuremath{\left(#1\right)}}
\newcommand{\PhaseA}[1]{\ensuremath{\mathcal{A}^{#1}}}
\newcommand{\PhaseB}[1]{\ensuremath{\mathcal{B}^{#1}}}
\newcommand{\PhaseC}[1]{\ensuremath{\mathcal{C}^{#1}}}
\newcommand{\PhaseText}{\ensuremath{\mathcal{T}_{ext}}}
\newcommand{\PhaseTmiddle}{\ensuremath{\mathcal{T}_{\mathit{middle}}}}
\newcommand{\PhaseSEC}{\ensuremath{\mathcal{T}_{\mathit{SEC}}}}

\newcommand{\Conf}{\ensuremath{\mathbf{\mathit{Conf}}}}

\subsection{Algorithm Description}

We propose in the following an algorithm that solves the SUIG problem. The moves of our algorithm depend on how many occupied locations (also called \emph{points} in the sequel) are seen by a robot and whether they are aligned or not. If two points are seen, we further divide the set of configurations depending on the distance between those points. More precisely, we partition the set $\Conf$ as follows:

\[
\Conf = \PhaseSEC \cup \PhaseText \cup \bigcup_{k\geq 0}\left(\PhaseA{k}\cup \PhaseB{k}\cup \PhaseC{k}_1\cup \PhaseC{k}_2\cup \PhaseC{k}_3 \right)
\]
\begin{itemize}
    \item $\PhaseText(p)$ (or simply $\PhaseText$ when $p$ is arbitrary) denotes the set of configurations where all robots are aligned and there is a single extremal robot $p$ at distance $e$ from the closest robot, with $\frac{e}{d}\in \left\{ \frac{1}{9}, \frac{1}{8}, \frac{9}{80}, \frac{10}{81}, \frac{10}{18}, \frac{9}{16}, \frac{80}{81} \right\}$, where $d$ denotes the distance between the two extremal robots\footnote{The fractions are a consequence of two moves defined in the algorithm ($\move{\frac{1}{9}}$ and  $\move{\frac{1}{8}}$). Other fractions (for the moves and thus in $\PhaseText$) could have been used, as long as the different combinations of moves can be distinguished in the obtained configuration.}. 
    If two such points exist, the configuration is not in $\PhaseText(p)$. 
\item $\PhaseSEC$ denotes the set of configurations that do not consist of only two points, and that are not in $\PhaseText$.
    \item The other sets of configurations form a partition of the set $\Conf_2$ of configurations consisting in only two points. When a configuration consists in only two points, we define the \emph{level} of a robot $r$ as the number $l_r$ such that $d_r \in [2^{-l_r}, 2^{-l_r+1})$, where $d_r$ is the distance between the two points seen by robot $r$. Then we define for all $k\geq 0$ 
    
    \[
    \begin{array}{ll}
       \PhaseA{k} &= \{C \in \Conf_2 \;|\; S_k\phantom{+2k+1}\; \leq l_r < S_k + k \}\\
       \PhaseB{k} &= \{C \in \Conf_2 \;|\; S_k+\phantom{2}k\phantom{+1}\; \leq l_r < S_k + 2k \}\\
       \PhaseC{k}_1 &= \{C \in \Conf_2 \;|\; S_k+2k\phantom{+1}\; \leq l_r < S_k + 2k+1 \} \\ 
       \PhaseC{k}_2 &= \{C \in \Conf_2 \;|\; S_k+2k+1 \leq l_r < S_k + 2k+2 \}\\
       \PhaseC{k}_3 &= \{C \in \Conf_2 \;|\;  S_k+2k+2 \leq l_r < S_k + 2k+3 \}
    \end{array}
    \]
    Where $S_k = k(k+2)$, so that $S_{k+1} = S_k + 2k + 3$ and all the levels $\geq 0$ are considered. Also, for simplicity, instead of being the empty set, we fix by convention $\PhaseA{0} = \{C \in \Conf_2 \;|\; l_r < 1 \}$. Doing so, the infinite sequence $\left(\PhaseA{k}\cup \PhaseB{k}\cup \PhaseC{k}_1\cup \PhaseC{k}_2\cup \PhaseC{k}_3\right)_{k\geq 0}$ is a partition of $\Conf_2$. One can notice that the deduced partition depends on the observing robot $r$, \ie, two different robots may see a configuration in $\Conf_2$ in different subsets at the same time, since they may have different distance units.
\end{itemize}

When the configuration consists of two points at distance $d$, we say a robot performs move $\move{e}$ if it is ordered to move toward the other point for distance $e \times d$.
We are now ready to describe our algorithm called SUIG for Stand-up Indulgent Gathering. The pseudo-code is given in Algorithm~\ref{algo:SUIG}, and a visual representation is given in Figure~\ref{fig:algo 1}.
When robots occupy only two points (\ie, when the configuration is in $\Conf_2$), a robot's action depends on its own coordinate system. More precisely, a robot $r$ orients the line passing through both occupied points using its coordinate system. If the other point (where the robot is not located) is located on the East or on the North (in case no point is at the East of the other), then $r$ considers it is on the left of the other point, otherwise $r$ sees itself on the right of the other point. In the illustration, we orient the line with an arrow, and present the moves using bent arrows, depending on where the robot sees itself, on the right or on the left. In the pseudo-code, we write

\[
\mathit{left}\rightarrow \texttt{Move1} ; \mathit{right}\rightarrow \texttt{Move2}
\] 
to notify that a robot that sees itself on the left, resp. on the right, should execute \texttt{Move1}, resp. \texttt{Move2}.

\vspace{0.6cm}
\begin{algorithm}[H]
\lIf{$C \in \PhaseA{k}$}{
      \move{\frac{1}{2}}
}
\lIf{$C \in \PhaseB{k}$}{
       \emph{left} $\rightarrow$ \move{\frac{1}{9}} ;
       \emph{right} $\rightarrow$ \move{\frac{1}{10}} 
}
\lIf{$C \in \PhaseC{k}_1$}{
       \move{\frac{1}{2}}
}
\lIf{$C \in \PhaseC{k}_2$}{
       \emph{left} $\rightarrow$ \move{1} ;
       \emph{right} $\rightarrow$ \move{\frac{1}{2}} 
}
\lIf{$C \in \PhaseC{k}_3$}{
       \emph{left} $\rightarrow$ \move{\frac{1}{2}} ;
       \emph{right} $\rightarrow$ \move{1} 
}
\lIf{$C \in \PhaseText(p)$}{
    Move to $p$
}
\lElse{
    Move to the center of the smallest enclosing circle.
}
 \caption{SUIG Algorithm}\label{algo:SUIG}
\end{algorithm}

\begin{figure}[h]
    \centering
    \begin{tabularx}{\textwidth}{lc}
       Phase \PhaseA{k}: & 
\multicolumn{1}{m{8cm}}{
\begin{tikzpicture}
\draw[->] (-1,0) -- (3,0);
\node[robot] (a) at (0,0) {};
\node[robot] (b) at (2,0) {};
\node[dest] (m) at (1,0) {};

\path (a) edge[bend left=60,->,>=stealth'] node [left] {} (m);
\path (b) edge[bend right=60,->,>=stealth'] node [left] {} (m);

\draw (0.5-0.06,-0.1) -- (0.5+0.06,0.1);
\draw (1.5-0.06,-0.1) -- (1.5+0.06,0.1);
\end{tikzpicture}
}
        \\
\makecell[l]{Phase \PhaseB{k}:\\
the left robot performs $\move{\frac{1}{9}}$\\
the right robot performs $\move{\frac{1}{10}}$}
& 
\multicolumn{1}{m{8cm}}{
\begin{tikzpicture}
\draw[->] (-1,0) -- (3,0);
\node[robot] (a) at (0,0) {};
\node[robot] (b) at (2,0) {};
\node[dest] (m1) at (2/10,0) {};
\node[dest] (m2) at (2-2/6,0) {};

\path (a) edge[bend left=80,->,>=stealth'] node [left] {} (m1);
\path (b) edge[bend right=80,->,>=stealth'] node [left] {} (m2);
\end{tikzpicture}
}
\\
Phase $\PhaseC{k}_1$: &
\multicolumn{1}{m{8cm}}{
\begin{tikzpicture}
\draw[->] (-1,0) -- (3,0);
\node[robot] (a) at (0,0) {};
\node[robot] (b) at (2,0) {};
\node[dest] (m) at (1,0) {};

\path (a) edge[bend left=60,->,>=stealth'] node [left] {} (m);
\path (b) edge[bend right=-60,->,>=stealth'] node [left] {} (m);

\draw (0.5-0.06,-0.1) -- (0.5+0.06,0.1);
\draw (1.5-0.06,-0.1) -- (1.5+0.06,0.1);
\end{tikzpicture}
}
\\
Phase $\PhaseC{k}_2$: &
\multicolumn{1}{m{8cm}}{
\begin{tikzpicture}
\draw[->] (-1,0) -- (3,0);
\node[robot] (a) at (0,0) {};
\node[robot] (b) at (2,0) {};
\node[dest] (m) at (1,0) {};

\path (a) edge[bend left=60,->,>=stealth'] node [left] {} (b);
\path (b) edge[bend right=-60,->,>=stealth'] node [left] {} (m);

\draw (0.5-0.06,-0.1) -- (0.5+0.06,0.1);
\draw (1.5-0.06,-0.1) -- (1.5+0.06,0.1);
\end{tikzpicture}
}
\\
Phase $\PhaseC{k}_3$: &
\multicolumn{1}{m{8cm}}{
\begin{tikzpicture}
\draw[->] (-1,0) -- (3,0);
\node[robot] (a) at (0,0) {};
\node[robot] (b) at (2,0) {};
\node[dest] (m) at (1,0) {};

\path (a) edge[bend left=-60,->,>=stealth'] node [left] {} (m);
\path (b) edge[bend right=60,->,>=stealth'] node [left] {} (a);

\draw (0.5-0.06,-0.1) -- (0.5+0.06,0.1);
\draw (1.5-0.06,-0.1) -- (1.5+0.06,0.1);
\end{tikzpicture}
}\\
\makecell[l]{Phase $\PhaseText(p)$\\
$\frac{e}{d}\in \left\{ \frac{1}{9}, \frac{1}{8}, \frac{9}{80}, \frac{10}{81}, \frac{10}{18}, \frac{9}{16}, \frac{80}{81} \right\}$
}&
\multicolumn{1}{m{8cm}}{
\begin{tikzpicture}
\draw[] (-1,0) -- (5,0);
\node[robot] (a2) at (0,0) {};
\node[robot] (r1) at (1,0) {};
\node[robot] (r2) at (1.7,0) {};
\node[robot] (r3) at (0.3,0) {};
\node[robot] (b2) at (2.6,0) {};
\node[robot,label={45:$p$}] (c2) at (4,0) {};

\draw[<->] (0,-0.5) -- (4,-0.5);
\node[] (e) at (2,-0.7) {$d$};

\path (a2) edge[bend left=40,->,>=stealth'] node [left] {} (c2);
\path (b2) edge[bend left=40,->,>=stealth'] node [left] {} (c2);
\path (r1) edge[bend left=40,->,>=stealth'] node [left] {} (c2);
\path (r2) edge[bend left=40,->,>=stealth'] node [left] {} (c2);
\path (r3) edge[bend left=40,->,>=stealth'] node [left] {} (c2);

\node[] (e) at (3.2,-0.2) {$e$};
\end{tikzpicture}}\\
Phase $\PhaseSEC$: & 
\multicolumn{1}{m{7cm}}{
\begin{tikzpicture}
\filldraw[color=black!60, fill=white, dashed](0,0) circle (1.3);
\node[robot] (a) at (-1.3,0) {};
\node[robot] (b) at (20:1.3) {};
\node[robot] (c) at (-40:1.3) {};
\node[robot] (r) at (140:0.8) {};
\node[robot] (r2) at (-100:1) {};
\node[dest] (d) at (0,0) {};

\path (a) edge[->,>=stealth'] node [left] {} (d);
\path (b) edge[->,>=stealth'] node [left] {} (d);
\path (c) edge[->,>=stealth'] node [left] {} (d);
\path (r) edge[->,>=stealth'] node [left] {} (d);
\path (r2) edge[->,>=stealth'] node [left] {} (d);
\end{tikzpicture} 
}
\end{tabularx}

    \caption{Algorithm \ref{algo:SUIG}.}
    \label{fig:algo 1}
\end{figure}
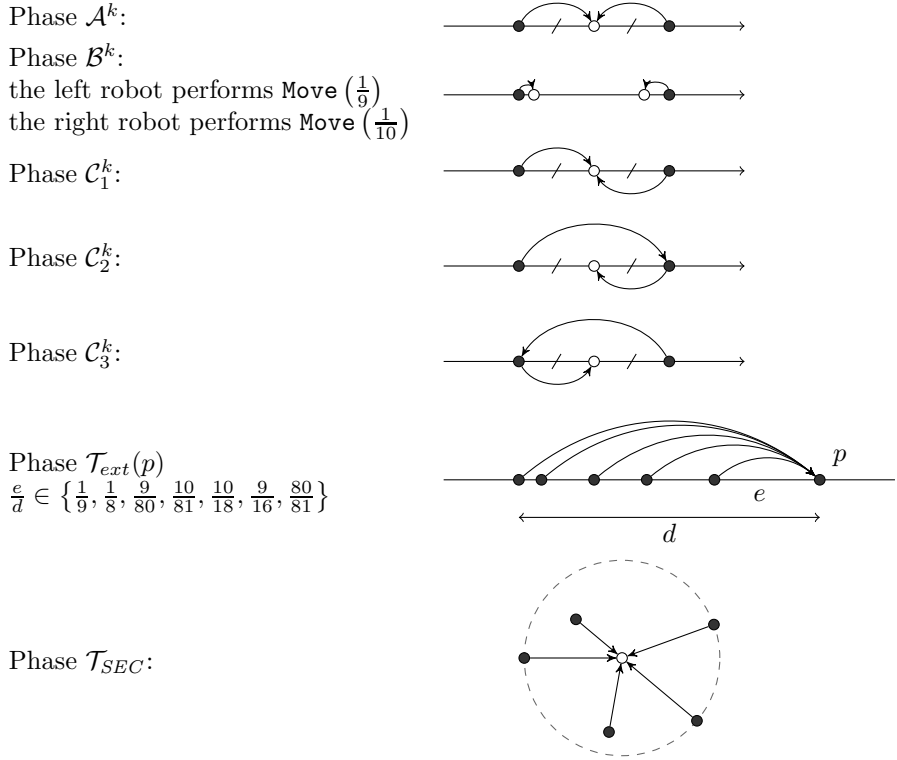

\subsection{Proof of Correctness} We prove in the following the correctness of our solution. 

\begin{lemma}\label{lem:movement are eventually rigid}
If robots do not gather, then they eventually remain at a distance at most $\delta$ from each other.
\end{lemma}

\begin{proof}
If the robots do not gather at the end of a round $t$, then there are at least two occupied locations. Let $D(t)$ be the diameter of the Smallest Enclosing Circle (SEC) at time $t$. Assume for the purpose of contradiction that for all $t\geq 0$, $D(t) > \delta$. Since robots never move outside of the convex hull they form, the function $t\mapsto D(t)$ is non-increasing, hence converges to a value $B\geq\delta$, and there is a time $T$ such that, 

\begin{equation}\label{eq:robots remains far away}
  \forall t>T,\qquad D(t) \in [B, B+\delta/10).   
\end{equation}

Consider any such time $t > T$. By our algorithm, the robots either move to the center of the Smallest Enclosing Circle (SEC), move to an extremity (when in phase $\PhaseText$), or perform $\move{e}$, $e\in \{\frac{1}{10}, \frac{1}{9}, \frac{1}{2}, 1\}$. 
We consider different cases depending on what moves robots execute in configuration 
$C(t)$.
\begin{itemize}
    \item \textbf{Case A:} If robots are in phase $\PhaseText(p)$, then the configuration is not symmetric and all robots have the same target, point $p$. Let $q \ne p$ be the other extremity of the segment hosting all robots. If no robot is crashed at $q$, then the diameter of the SEC decreases by at least $\delta/10$ (in the worst case, the crashed robot is at distance $9D(t)/10$ from $p$), which contradicts Relation~\ref{eq:robots remains far away}.
    If there is a crashed robot at $q$, then, in configuration $C(t+1)$, no correct robot is collocated with the crashed robot at $q$ and either case $B$, $C$ or $D$ applies to reach a contradiction.

    \item \textbf{Case B:}
If robots are in phase $\PhaseSEC$, then all the robots (except maybe one crashed robot) move towards the center of the SEC. If the diameter of the SEC is at most $2\delta$, each robot reaches its target and the SEC at time $t+1$ is at least halved, so decreases by at least $\delta/2$ (because the radius at time $t$ is $D(t)\geq\delta$). If the diameter of the SEC is greater than $2\delta$, as illustrated in Figure~\ref{fig:case B}, the diameter of the SEC decreases by at least $\delta$ (it decreases by $2\delta$ if there is no crash).
Both cases contradict Relation~\ref{eq:robots remains far away}.

In the remaining cases, $C(t)$ consists of two points.

\begin{figure}
    \centering
\begin{tikzpicture}
\filldraw[color=black!60, fill=white, dashed](0,0) circle (1.3);
\filldraw[color=black!60, fill=white, dashed](-0.2,0) circle (1.1);

\tikzset{
    cross/.pic = {
    \draw[rotate = 45] (-#1,0) -- (#1,0);
    \draw[rotate = 45] (0,-#1) -- (0, #1);
    }
}
\node[] (a) at (-1.3,0) {};
\path (-1.3,0) pic[red] {cross=4pt};
\node[robot] (b) at (20:1.3) {};
\node[robot] (c) at (-40:1.3) {};
\node[robot] (r) at (140:0.8) {};
\node[robot] (r2) at (-100:1) {};
\node[circle, draw=black, fill=white, inner sep=0.002cm] (d) at (0,0) {};

\path (b) edge[->,>=stealth'] node [left] {} (20:0.9);
\path (c) edge[->,>=stealth'] node [left] {} (-40:0.9);
\path (r) edge[->,>=stealth'] node [left] {} (140:0.4);
\path (r2) edge[->,>=stealth'] node [left] {} (-100:0.6);

\path (0:0.9) edge[<->,>=stealth'] (0:1.3);
\node[] at (-15:1.1) {$\delta$};

\end{tikzpicture} 
    \caption{Case B: Reduction of the diameter of the SEC by $\delta$ when each robot, except the crashed one, move towards the center of the SEC}
    \label{fig:case B}
\end{figure}
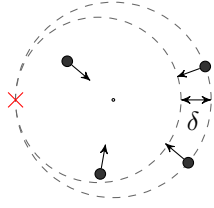 
    \item \textbf{Case C:}
If robots are not in phase $\PhaseText\cup\PhaseSEC$ and no robot executes $\move{1}$, then the robots at the two extremities execute $\move{\frac{1}{10}}$, $\move{\frac{1}{8}}$ or $\move{\frac{1}{2}}$ and the diameter of the SEC (which is the distance between the two extremities) decreases by at least $\min(\delta, d/10) \geq \delta/10$, which contradicts Relation~\ref{eq:robots remains far away}.

    \item \textbf{Case D:}
If robots are not in phase $\PhaseText\cup\PhaseSEC$, some robots execute $\move{1}$, and if there is a crashed robot with no collocated correct robots. In this case, all the correct robots move towards the crashed robot a distance at least $\delta/10$ (because some robots might be in different phases), so the diameter of the SEC decreases by at least $\delta/10$, which contradicts Relation~\ref{eq:robots remains far away}.

    \item \textbf{Case E:}
If robots are not in phase $\PhaseText\cup\PhaseSEC$, some robots execute $\move{1}$, and if there is no crashed robot or a crashed robot collocated with correct robots.
If there is a robot that executes a move different from $\move{1}$, then the next configuration $C(t+1)$ does not consist of only two points and a previous case applies to reach a contradiction.
Now assume all robots execute $\move{1}$ and configuration $C(t+1)$ consists in only two points. Since $D(t+1)$ is still in the interval $[B, B + \frac{\delta}{10})$, then the moving robots must have changed their side on the segment. So if the level of robot has not changed, it executes $\move{\frac{1}{2}}$ at time $t+1$.
At time $t+1$, either a previous case applies to reach a contradiction (if a robot executes a move different to $\move{1}$) or the same case applies.
When the same case applies, as illustrated in Figure~\ref{fig:case E}, all the robots have increased their level and execute again $\move{1}$ at time $t+1$. If the same case applies at time $t+2$, this means that their level increased two times and this occurs only if the distance between the two points has been halved at least once, which contradicts Relation~\ref{eq:robots remains far away}. So a previous case applies and a contradiction is reached as well.

\begin{figure}
    \centering
\begin{tikzpicture}

\draw[] (-0.5,0) -- (3.5,0);
\node[robot] (a) at (0,0) {};
\node[robot] (b) at (3,0) {};

\node[dest] (d1) at (2.85,0) {};
\node[dest] (d2) at (0.15,0) {};

\path (a) edge[bend left=-60,->,>=stealth'] node [left] {} (d1);
\path (b) edge[bend right=60,->,>=stealth'] node [left] {} (d2);

\draw[->,>=stealth'] (1.5,-1.1) -- (1.5,-1.5);

\begin{scope}[yshift=-2.5cm]
\draw[] (-0.5,0) -- (3.5,0);
\node[robot] (a) at (0.15,0) {};
\node[robot] (b) at (2.85,0) {};

\node[dest] (d1) at (2.7,0) {};
\node[dest] (d2) at (0.3,0) {};

\path (a) edge[bend left=-60,->,>=stealth'] node [left] {} (d1);
\path (b) edge[bend right=60,->,>=stealth'] node [left] {} (d2);

\draw[decoration={
            text along path,
            text={|\footnotesize| Impossible},
            text align={center},raise=0.2cm},decorate] (-1,-1.5) --  (0.5,-1.1);
\draw[<-,>=stealth'] (-1,-1.5) --  (0.5,-1.1);

\draw[<-,>=stealth'] (4,-1.5) --  (2.5,-1.1);

\end{scope}

\begin{scope}[yshift=-5.0cm, xshift=-3cm]
\draw[] (-0.5,0) -- (3.5,0);
\node[robot] (a) at (0.3,0) {};
\node[robot] (b) at (2.7,0) {};

\node[dest] (d1) at (2.55,0) {};
\node[dest] (d2) at (0.45,0) {};

\path (a) edge[bend left=-60,->,>=stealth'] node [left] {} (d1);
\path (b) edge[bend right=60,->,>=stealth'] node [left] {} (d2);

\end{scope}

\begin{scope}[yshift=-5.0cm, xshift=3cm]
\draw[] (-0.5,0) -- (3.5,0);
\node[robot] (a) at (0.3,0) {};
\node[robot] (b) at (2.7,0) {};

\path (a) edge[bend right=-60,->,>=stealth', dotted] node [left] {} (0.9,0);
\path (a) edge[bend right=-60,->,>=stealth', dotted] node [left] {} (1.2,0);
\path (a) edge[bend left=-60,->,>=stealth', dotted] node [left] {} (1.5,0);
\path (b) edge[bend left=60,->,>=stealth', dotted] node [left] {} (2.1,0);
\path (b) edge[bend right=60,->,>=stealth', dotted] node [left] {} (1.5,0);

\end{scope}

\end{tikzpicture} 
    \caption{Case E: If robots nearly exchange their position and increase their level, it is possible they all execute $\move{1}$ again, but this  cannot occur one more time, because the distance between the two points must be at least halved in order for the robot to see their level increased by 2.}
    \label{fig:case E}
\end{figure}
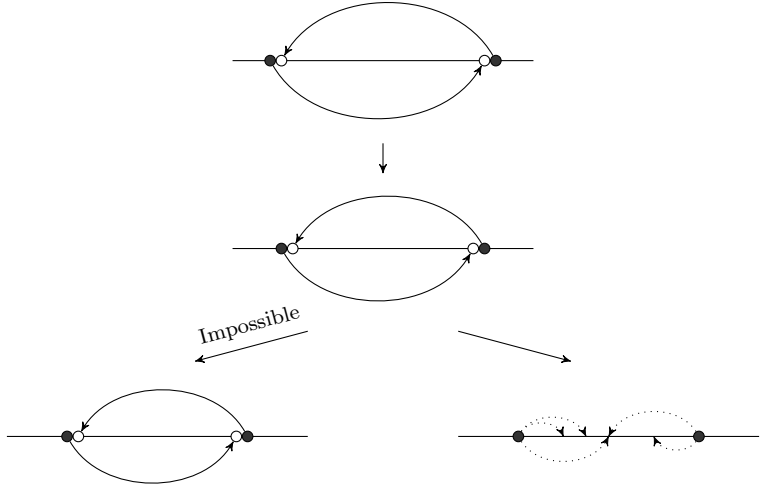 \end{itemize}
\end{proof}

By Lemma~\ref{lem:movement are eventually rigid}, we can assume without loss of generality that robot movements are rigid. The next Lemma and its Corollary show how we use the fact that phases contain an increasing number of levels ($\PhaseA{k+1}$ has 1 more level than $\PhaseA{k}$). Lemma~\ref{lem:increase of level is at most 7} states that robots cannot ``jump up'' more than 7 levels, and Corollary~\ref{lem:all robots are in phase A (resp. B)} uses this fact to state that there is a time when all the robots are in the same phase $\PhaseA{k}$, resp. $\PhaseB{k}$, for a $k$ large enough.

In the sequel, the \emph{lowest robot} denotes the robot with the lowest level. $\Delta_{level}$ is the maximal difference between any two robot levels, \ie, $\Delta_{level} = \max_{r, r'}|l_r - l_{r'}|$. The difference may depend on the configuration but is upper-bounded. Indeed, if the lowest robot has a unit distance less than $2^c$ times larger than the highest robot, then the difference in level is at most $c$. 

\begin{lemma}\label{lem:increase of level is at most 7}
In any execution, either all robots gather, or there exists a level $l$ such that for any level $l'\geq l$, the robots reach a configuration consisting in two points and the lowest robot has a level in the interval $[l', l' + 7]$.
\end{lemma}

\begin{proof}
Assume the robots do not gather, from Lemma~\ref{lem:movement are eventually rigid} we can assume rigid movements from a given time $t_0$. If $C(t_0)$ consists in two points, let $l$ be the level of the lowest robot in $C(t_0)$. Otherwise, all robots are dictated to move towards the same destination (the phase is $\PhaseSEC$ or $\PhaseText$) so either all robots gather in $C(t_0+1)$, or $C(t_0+1)$ consists of two points (if there is a crashed robot). In the latter case, let $l$ be the level of the lowest robot in $C(t_0+1)$.

To prove the Lemma, it is sufficient to prove the following claim: if a configuration $C(t)$ at time $t$ consists in two points at distance $d$, then either $C(t+1)$ or $C(t+2)$ consists in two points at distance $d' \geq \frac{d}{2^7}$, or all robots gather in $C(t+1)$ or $C(t+2)$. Indeed, this implies that the level between two consecutive configurations forming two points cannot increase by more than 7. This implies that, for any $l'\geq l$ we eventually reach a configuration where the lowest robot has a level in the interval $[l', l' + 7]$, or all robots gather.  

We now prove the claim. 
After executing the algorithm from $C(t)$, either all robots gather, or $C(t+1)$ consists in two points, or in three or more points. If $C(t+1)$ consists in two points, the distance between any two points is at least $(\frac{1}{9} - \frac{1}{10})d = \frac{1}{90}d \geq \frac{1}{2^7}d$. So, the level increases by at most 7 between $C(t)$ and $C(t+1)$.

If $C(t+1)$ consists in 3 points or more, the minimal interval containing three points is at least $\frac{1}{10}d$, so that the extremities (recall the robots are aligned) are at distance $d' \geq \frac{1}{10}d$.
In $C(t+1)$ the robots all move towards the same destination. If there is no crashed robot or if the crashed robot is already at the destination, then the next configuration $C(t+2)$ is gathered. Otherwise, since the crashed robot is located at an extremity and the destination is either an extremity or the middle, then $C(t+2)$ consists of two points at distance $d''\geq d'/2$. Hence we have $d'' \geq \frac{1}{20}d \geq \frac{1}{2^7}d$.

So the claim is proved and the Lemma follows.
\end{proof}

\begin{corollary}\label{lem:all robots are in phase A (resp. B)}
For any $k\geq 0$, eventually either all robots gather, or the robots are all in phase \PhaseA{k'} with $k'\geq k$. The same is true with phase \PhaseB{}.
\end{corollary}

\begin{proof}
Let $l$ be defined by Lemma~\ref{lem:increase of level is at most 7}, and let $k > \Delta_{level} + 7$ such that $S_k \geq l$.
We can now apply Lemma~\ref{lem:increase of level is at most 7} with $l' = S_{k'}$, and $k' \geq k$ (clearly we have $l' > l$ so the lemma applies). Thus, we know that either the robots eventually gather, or reach a configuration where the lower robot has a level in the interval $[S_{k'}, S_{k'} + 7]$. Since $k'> \Delta_{level} + 7$, all the robots have a level in the interval $[S_{k'}, S_{k'} + k' - 1]$, and they all are in phase \PhaseA{k'}.

A similar proof can be done by replacing \PhaseA{k'} with \PhaseB{k'}.
\end{proof}

We are now ready to prove the correctness of our algorithm. We first prove that robots gather if no robot crashes (Lemma~\ref{lem:correctness, no robots crash}), and then when robots may crash at the same location (Lemma~\ref{lem:correctness, a robot crashes}). For simplicity, we consider in the sequel, without loss of generality, that the movements are rigid (using Lemma~\ref{lem:movement are eventually rigid}).

\begin{lemma}\label{lem:correctness, no robots crash}
If no robot crashes, all robots eventually gather at the same point.
\end{lemma}

\begin{proof}
By Corollary~\ref{lem:all robots are in phase A (resp. B)}, either the robots gather or all enter phase \PhaseA{k} for some $k>0$. Hence, after one more round all robots move to the middle and gather.
\end{proof}

\begin{lemma}\label{lem:correctness, a robot crashes}
If robots crash at the same location, all robots eventually gather at the same point.
\end{lemma}

\begin{proof}
Assume some robots crash at the same location. By Corollary~\ref{lem:all robots are in phase A (resp. B)}, either the robots gather, or all enter phase \PhaseB{k} for some $k>0$. Let $d$ denote the distance between the two points.

\textbf{If some correct robots are located at the crashed location}, then after one round, at least three points are formed. In the obtained configuration, all the robots remain on the same line, and the new distance $d'$ between the two farthest robots depends on which move the robots performed (move to $\frac{1}{10}$, or move to $\frac{1}{9}$). In case 1, if at least one robot at the opposite of the crashed location performed a $\frac{1}{10}$ move, then $d'=\frac{9}{10}d$, otherwise, in case 2, $d'=\frac{8}{9}d$.
Also, all the correct robots that were on the crashed location are now at distance $\frac{1}{10}d$ or $\frac{1}{9}d$ from the crashed location (again depending on what move they performed). So depending on the new value of $d'$, the crashed location is at the extremity of the line, and its distance from the closest correct robot is either in $\left\{\frac{1}{9}d', \frac{10}{81}d'\right\}$ in case 1 (\emph{e.g.}, $\frac{1}{9}d'$ is obtained by robots that performed a $\frac{1}{10}$ move, \ie, $\frac{1}{10}d=\frac{1}{10}\frac{10}{9}d'=\frac{1}{9}d'$), or in $\left\{\frac{9}{80}d', \frac{1}{8}d'\right\}$ in case 2.

Note that the robots at the other extremity of the line have a distance from the closest robot either $\frac{1}{81}d'$ (if there were at least two robots that moved differently) or is greater than $\frac{1}{2}d'$ (otherwise).
All the possible cases are represented in Figure~\ref{fig:move in phase B}.

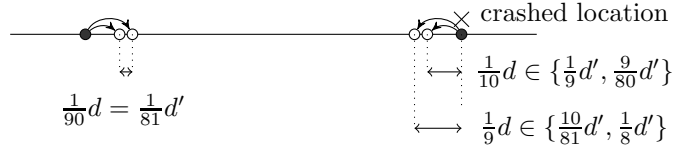
\begin{figure}
    \centering
    
    \begin{tikzpicture}[scale=5]
\draw[] (-0.2,0) -- (1.2,0);
\node[robot] (a) at (0,0) {};
\node[robot] (b) at (1,0) {};
\node[dest] (a1) at (1/8,0) {};
\node[dest] (a2) at (1/11,0) {};
\node[dest] (b1) at (1-1/8,0) {};
\node[dest] (b2) at (1-1/11,0) {};

\path (a) edge[bend left=60,->,>=stealth'] node [left] {} (a1);
\path (a) edge[bend left=40,->,>=stealth'] node [left] {} (a2);
\path (b) edge[bend right=60,->,>=stealth'] node [left] {} (b1);
\path (b) edge[bend right=40,->,>=stealth'] node [left] {} (b2);

\draw[] (1-0.02,0.02) -- (1+0.02,0.02+0.04);
\draw[] (1-0.02,0.02+0.04) -- (1+0.02,0.02);
\node at (1.3,0.06) {crashed location};

\draw[dotted] (1/11,0) -- (1/11,-0.1);
\draw[dotted] (1/8,0) -- (1/8,-0.1);
\draw[<->] (1/8,-0.1) -- (1/11,-0.1);
\node at (0.1, -0.2) {$\frac{1}{90}d = \frac{1}{81}d'$};

\draw[dotted] (1,0) -- (1,-0.2);
\draw[dotted] (1-1/11,0) -- (1-1/11,-0.1);
\draw[<->] (1,-0.1) -- (1-1/11,-0.1);
\draw[dotted] (1-1/8,0) -- (1-1/8,-0.25);
\draw[<->] (1,-0.25) -- (1-1/8,-0.25);

\node at (1.3, -0.1) {$\frac{1}{10}d\in \{ \frac{1}{9}d', \frac{9}{80}d' \}$};
\node at (1.3, -0.25) {$\frac{1}{9}d\in \{ \frac{10}{81}d', \frac{1}{8}d' \}$};

\end{tikzpicture}
    \caption{The possible destinations of robots after executing a move in phase B, assuming a single crashed location (on the right in this example). All the white nodes represent the possible destinations of the correct robots. We can see that regardless of where the robots end up, all the robots can detect where is the crashed location, using the ratio of the distance from an extremity to the closest robot.}
    \label{fig:move in phase B}
\end{figure}

In any case, all the robots can detect where the crashed location is located. The crashed location is at the extremity of the line such that the closest robot is at distance in $\{\frac{1}{9}d', \frac{9}{80}d', \frac{10}{81}d', \frac{1}{8}d'\}$, where $d'$ is the length of the line (recall that the ratio does not depend on the value of $d'$). Thus, correct robots gather at the crashed location after one more round by performing move $\PhaseText$.

\textbf{If no correct robot is located at the crashed location, and some robots have different orientations}, the same thing happens. The obtained configuration contains three points, two of them are at distance $\frac{1}{81}d'$, and the crashed location is at distance $\frac{80}{81}d'$ from the closest correct robot, so that correct robots gather at the crashed location after one more round by performing move $\PhaseText$.

\textbf{If no robot is located at the crashed location, and all the robots have the same orientation}, then the correct robots come closer to the crashed location every round until one robot enters phase $\PhaseC{k}_1$. If some correct robots are still in phase $\PhaseB{k}$, then, similarly to the previous cases, the obtained configuration contains 3 or 4 points, and all correct robots can detect where the crashed location is, as it is the extremity with distance from the closest robot in $\{\frac{10}{18}d', \frac{9}{16}d'\}$. Thus, correct robots gather at the crashed location after one more round by performing move $\PhaseText$.

If all correct robots enter phase $\PhaseC{k}_1$ at the same time, then after one more round they all enter phase $\PhaseC{k}_2$. If they are on the left of the crashed location, all correct robots gather at the crashed location, otherwise, they all move to the middle, enter phase $\PhaseC{k}_3$, and then move and gather at the crashed location.
\end{proof}

From Lemma~\ref{lem:correctness, no robots crash} and Lemma~\ref{lem:correctness, a robot crashes}, the main theorem follows.

\begin{theorem}
Algorithm~\ref{algo:SUIG} solves the SUIG problem in FSYNC, for any initial configuration, with $n\geq 1$ disoriented robots, without multiplicity detection, and with non-rigid movements.
\end{theorem}

\newcommand{\deltalevel}{\Delta_{\mathit{level}}}
The round complexity of our algorithm can only be bounded when there is no crash and when the movements are rigid. It depends on $l_{\min}$, the level of the lowest robot, and
on $\deltalevel$, the difference between the level of the largest robot and the level of the lowest level.

\begin{theorem}
If movements are rigid and there is no crash,
the round complexity of Algorithm~\ref{algo:SUIG} is \[
O\left(
\deltalevel^2
+
\sqrt{l_{\min}}
\right).
\]
\end{theorem}

\begin{proof}

If there is more than two occupied points, the gathering is achieved in one round. Otherwise, it depends on the level of the robots.

Case $(i)$: the level of the lowest robot is smaller than $S_{\Delta_{\mathit{level}}+7}$. We know that the number of rounds for the lowest robot to reach level $S_{\Delta_{\mathit{level}}+7}$, \ie, to reach phase $\PhaseA{\Delta_{\mathit{level}}+7}$, is in $O\left(\Delta_{\mathit{level}}^2\right)$ rounds. When this level is reached, then, robots gather in one round (they all execute $\move{\frac{1}{2}}$). 

Case $(ii)$: the level $l_{\min}$ of the lowest robot is between $S_{k}$ and $S_{k+1}$, with $k\geq \Delta_{\mathit{level}}+7$. 
In the worst case, the gathering is achieved one round after the lowest robot reaches level $S_{k+1}$ (\ie when all robots execute $\PhaseA{k + 1}$). So the round complexity in this case is $O(k) \subset O\left(\sqrt{l_{\min}}\right)$ (because $k$ is defined such that $S_k = k(k+2) \leq l_{\min}$). \end{proof}

In fact, when robots can crash, the adversary can crash a robot arbitrarily close to the other robots, so that the level of the other robots is greater that $S_k$ with $k$ arbitrarily large. To reach phase $\PhaseB{k+1}$, the number of round can be in $\Omega(k)$ with $k$ arbitrarily large. Similarly, if there is no crash but movement are not rigid, then the adversary can stop one robot arbitrarily close to the other robots, creating the same situation. In both cases, the previous theorem can be applied from this point (one can see that if there is a crash, but we start with two occupied node, the previous theorem still hold).

\section{Concluding Remarks}
\label{sec:conclusion}

We presented the first stand-up indulgent gathering (also known as strong fault-tolerant gathering in the literature) solution. Furthermore, our solution is self-stabilizing (the initial configuration may include multiplicity points, and even be bivalent), does not rely on extra assumptions such as multiplicity detection capacity, a common direction, orientation, or chirality, etc, and its running time is proportional to the maximum initial distance between robots. 

The very weak capacities of the robots we considered make the problem unsolvable in more relaxed execution models, such as SSYNC. However, it may be possible to solve SUIG in SSYNC with a number of additional assumptions: non-bivalent initial configurations (otherwise, the impossibility from Corollary~\ref{cor:impossibility of SUIG in SSYNC} applies), multiplicity detection capacity (otherwise, the impossibility for classical gathering applies~\cite{prencipe07tcs}), and persistent coordinate system (otherwise, the impossibility by Defago et al.~\cite{defago20dc} about SUIG applies). This open question is left for future research.

\bibliographystyle{plain}
\bibliography{biblio}

\end{document}